\documentclass[runningheads]{llncs}
\usepackage{graphicx}
\usepackage[colorinlistoftodos]{todonotes}
\usepackage{amsmath,amsfonts}
\usepackage{cleveref}
\newcommand{\V}{\mathbb{V}}
\newcommand{\agents}{\mathcal{A}}
\newcommand{\cho}{choose}
\newcommand{\dec}{decide}

\begin{document}
\title{A Sufficient Epistemic Condition for Solving Stabilizing Agreement}

\author{Giorgio Cignarale \inst{1}\orcidID{0000-0002-6779-4023} \and
Stephan Felber\inst{1}\orcidID{0009-0003-6576-1468} \and
Hugo Rincon Galeana\inst{1}\orcidID{0000-0002-8152-1275}}
\authorrunning{G. Cignarale, S. Felber, and H. Rincon Galeana}
\titlerunning{A Sufficient Epistemic Condition for Solving Stabilizing Agreement}
\institute{Embedded Computing Systems,
TU Wien,
Vienna, Austria\\
\email{giorgio.cignarale@tuwien.ac.at}\\
\email{stephan.felber@tuwien.ac.at}\\
\email{hugo.galeana@tuwien.ac.at}\\
}

\maketitle              

\begin{abstract}
In this paper we provide a first-ever epistemic formulation of stabilizing agreement, defined as the non-terminating variant of the well established consensus problem.
In stabilizing agreements, agents are given (possibly different) initial values, with the goal to eventually always decide on the same value. While agents are allowed to change their decisions finitely often, they are required to agree on the same value eventually. We capture these properties in temporal epistemic logic and we use the Runs and Systems framework to formally reason about stabilizing agreement problems. We then epistemically formalize the conditions for solving stabilizing agreement, and identify the knowledge that the agents acquire during any execution to choose a particular value under our system assumptions. This first formalization of a sufficient condition for solving stabilizing agreement sets the stage for a planned necessary and sufficient epistemic characterization of stabilizing agreement. 

\keywords{Distributed Systems, Epistemic Logic, Runs and Systems, Temporal Logic, Stabilizing Agreement}

\end{abstract}
\section{Introduction}
\label{sec:Intro}

We introduce the stabilizing agreement problem with a twist to a famous puzzle: the two kingdoms of Aldinga and Beluga have decided to combine their powers and merge their kingdoms, but haven't agreed under which name they will continue to write history. The kingdoms reside on two hills separated by a valley, where unfortunately a hungry dragon lives.

Fortunately, the dragon is getting slow and only ever eats at most one of the two messengers between the two kingdoms (but sometimes catches neither). Neither Aldinga nor Beluga are in a rush to agree on either of the names, but they both know that they should eventually start to use the same name as their respective heroic deeds are only impressive enough to leave a mark in the history books if they appear under the same name forever.

The situation of Aldinga and Beluga is derived from the two generals problem~\cite{LSP82} and not surprisingly the two kingdoms can also not solve consensus. Interestingly enough though,~\cite{sh2024stabilizing} showed that stabilizing agreement is indeed possible in this communication setting, highlighting that stabilizing agreement is a proper weakening of the \emph{terminating} consensus problem. Further, Aldinga and Beluga can also solve the non-byzantine version of the firing-rebels problem ~\cite{fire}, which itself is a non-simultaneous firing version of the firing squad ~\cite{charronbost_et_al:LIPIcs.STACS.2018.20} should they choose to fire on the dragon. This hints to a hierarchy of the three problems, where consensus as the strongest problem implies stabilizing agreement and the firing rebels, stabilizing agreement implies firing rebels, and firing rebels as the weakest of all three.
This conjecture is backed by the fact that the sufficient knowledge for stabilizing agreement shown in this paper implies the sufficient knowledge for firing rebels~\cite{fire}, and from stabilizing agreement being a non-terminating relaxation of consensus. 

Formally, we consider a \emph{distributed system} consisting of multiple processes, also known as \emph{agents}, that communicate and coordinate actions in order to solve some given problem, such as consensus~\cite{DS:concept_design,Lynch_DA}. Temporal epistemic logic proved to be an extremely successful framework for distributed system \cite{bookof4}, providing a number of crucial results such as the \emph{Knowledge of Precondition Principle} formulated by Moses~\cite{Mos15TARK} and the \emph{happened-before relation} formulated by Lamport~\cite{Lam78}.
Another important insight provided by epistemic logic is that the common knowledge operator (everybody knows that everybody knows etc.) is closely related to perfect coordination, which is required in consensus \cite{bookof4}. Maybe unsurprisingly, even the fault-tolerant version of the Firing Rebels problem requires an eventual common epistemic attitude, as shown in \cite{fire}.

Stabilizing problems relax the standard consensus in one crucial way: albeit the goal in both is to have all agents agreeing on a common value, stabilizing agreement does not require agents to terminate. This means that an agent can \textit{choose} a value in a certain state of the system, supposing that it is indeed the correct one, only to discover later on that another value is more likely to be the common value. This deceptively simple twist keeps the agents ever-guessing, as agents in general have no guarantee of knowing when they have obtained all values.

While stabilizing agreement has recently received increasing attention  \cite{charronbost2019minmax,sh2024stabilizing}, its epistemic formalization is still missing.
The goal of this paper is to fill this gap, in providing a sufficient condition for solving stabilizing agreement.
\vspace{-10pt}

\section{Epistemic Framework}
\label{sec:R&S}
\vspace{-10pt}
For our epistemic logic framework, we will use a minimalist framework based on the \textit{runs and systems} framework \cite{bookof4}, which integrates the possible world Kripke semantics of epistemic logic \cite{van2015handbook} with a temporal-component.

We fix a finite set $\agents=\{1,\ldots,n\}$ of agents with perfect recall, and a finite set of values $\V=\{1,\ldots,k\}$.
Each agent is assigned an initial value, formally represented by the atom $init_a(v)$ for $a \in \mathcal{A}$ and $v \in \V$. Each agent knows its own initial value but does not know the local values of other agents.
Each agent performs actions (according to its protocol), e.g., send messages. One of the actions that any agent can do is to \emph{choose} a value among the initial values learned during a run, formally represented by the atom $\cho_a(v)$ with $a \in \agents$ and $v \in \V$.
We assume that there exists a linearly ordered global time set $\mathbb{T}$, which we will assume for convenience to be $\mathbb{N}$ \footnote{It should be noted that other time sets such as $\mathbb{R}$ may be considered, but will not be considered within the scope of this paper.}. 

We also consider a set of global states $\mathcal{G}$, along with a collection of equivalence accessibility relations  $\mathcal{R}_a \subseteq \mathcal{G} \times \mathcal{G}$, indexed by each agent $a \in \agents$, reflecting global state indistinguishability relative to each agent. This accessibility relation implicitly reflects the agents' local views. Concretely, for $s_1, s_2 \in \mathcal{G}$ we say that $(s_1, s_2) \in \mathcal{R}_{a}$ iff agent $a$ cannot distinguish between world $s_1$ and world $s_2$.
A run $\sigma$ is a sequence of global states $\sigma(t)_{t \in \mathbb{N}} =
\sigma(0),\ldots,\sigma(t), \ldots$ of the system, and we denote the set of runs by $\Sigma$. Within the scope of this paper, we will assume that for any possible world $w \in \mathcal{G}$, there is a run $\sigma \in \Sigma$ and a time $t \in \mathbb{T}$ such that $w = \sigma(t)$.
\vspace{-10pt}
\section{Epistemic Modeling of Stabilizing Agreement}\label{sec:stabtask}

\vspace{-10pt}
 We consider a language $\mathcal{L}$, defined by the following grammar: \[\varphi ::=  \ p  \ | \ \neg\varphi \ | \  (\varphi \wedge \varphi) \ | \ K_a \varphi \ | \ \Diamond \varphi.\]
 
 We assume that $p \in Prop$ is the set of propositional atoms, in particular $Prop$ includes all the atoms previously discussed in \Cref{sec:R&S}, $a\in \agents$; derived Boolean connectives are defined in the usual way and we use the following abbreviations: $\Box \varphi := \neg \Diamond \neg \varphi$ (always) and $E_{\agents} \varphi:= \bigwedge_{a\in\agents}K_a \varphi$ (mutual knowledge).

We assume a valuation function $\pi: Prop \rightarrow 2^{\mathcal{G}}$ that determines the semantics. In particular, for a run $\sigma \in \Sigma$, time $t \in \mathbb{N}$, atomic proposition $p \in Prop$, agent $a \in \agents$ , and formula $\varphi$, $(\mathcal{I},\sigma,t) \models p$ iff $\sigma(t) \in \pi(p)$, and $(\mathcal{I},\sigma,t) \models K_a \varphi$ iff $(\mathcal{I},\sigma',t') \models \varphi$ for any $\sigma' \in \Sigma$ and $t' \in \mathbb{N}$ such that $(\sigma(t) , \sigma'(t')) \in \mathcal{R}_a$; and $(\mathcal{I},\sigma,t) \models \Diamond \varphi$ iff $(\mathcal{I},\sigma,t') \models \varphi $ for some $t' \geq t$.
A formula $\varphi$ is valid in $\mathcal{I}$ during a run $\sigma \in \Sigma$, written $(\mathcal{I}, \sigma) \models \varphi$ iff $(\mathcal{I}, \sigma,t) \models \varphi$ for any $t \in \mathbb{T}$.
A formula $\varphi$ is valid in $\mathcal{I}$, written $\mathcal{I} \models \varphi$, iff $(\mathcal{I},\sigma,t) \models \varphi$ for all $\sigma \in \Sigma$ and $t \in \mathbb{N}$. 
We assume input values to be stable during a run, meaning that for any run $\sigma \in \Sigma$, $\mathcal{I}, \sigma \models init_a(v)$ or $\mathcal{I}, \sigma \models \neg init_a(v)$.

Since we focus on agreement tasks, we are interested in knowledge formulas that are relative only to input values: 

\begin{definition}[Primitive value formula]
    We denote the set of \emph{primitive value formulas}, $\Phi:= \{ \bigwedge_{a \in X } init_a(v_a) \; \vert \; \emptyset \neq X \subseteq \agents, \; v_a \in \V \}$
\end{definition}

Since $\mathcal{A}$ and $\V$ are finite, also $\Phi$ is finite.

We write $K_a \varphi$ for some $\varphi \in \Phi$ to denote that agent $a$ knows (for example, by receiving messages containing the values) the initial values represented in $\varphi$. Usually we only want to reason about the 'largest' set of initial values agent $a$ is aware of: 

\begin{definition}[Current primitive knowledge]\label{def:phibar} \footnote{It should be noted that \cref{def:phibar}, \cref{def:phibarlim}, \cref{def:phistar}, and \cref{def:phistarlim} are not unique formulas, but rather an equivalence class of propositional formulas that may differ via permutation.}
    We denote by $\overline{\varphi}_{(a,\sigma,t)}$, the current primitive knowledge of $a$ at a run $\sigma$ and a time $t$ if, for $\overline{\varphi}_{(a,\sigma,t)}\in \Phi$, $(\mathcal{I}, \sigma, t) \models K_a \overline{\varphi}_{(a,\sigma,t)}$, and for any formula $\varphi' \in \Phi$ such that $(\mathcal{I}, \sigma, t) \models K_a \varphi'$, then $\overline{\varphi}_{(a,\sigma,t)}\rightarrow \varphi'$.
\end{definition}

We also define the limit of a current primitive knowledge:
\begin{definition}[Primitive knowledge limit]\label{def:phibarlim}
    Let $\sigma \in \Sigma$ be a run, $a \in \mathcal{A}$, and $t \in \mathbb{T}$. We define the limit primitive knowledge of $a$ at $\sigma$, denoted by $\overline{\varphi}_{(a,\sigma)} \in \Phi$ as the strongest primitive value formula that $a$ will know at $\sigma$. Formally, this means that $(\mathcal{I},\sigma) \models \Diamond K_a \overline{\varphi}_{(a,\sigma)}$ and for any $\varphi' \in \Phi$ such that $(\mathcal{I}, \sigma, t) \models K_a \varphi' $, it follows that $\overline{\varphi}_{(a,\sigma)} \rightarrow \varphi'$.
\end{definition}

Assuming that a process has achieved knowledge of mutual primitive knowledge, then we can also define the strongest mutually known primitive knowledge relative to a process:
\begin{definition}[Current mutually-known primitive knowledge] \label{def:phistar}
    Let $a \in \mathcal{A}, \sigma \in \Sigma, t\in \mathbb{T}$ such that $(\mathcal{I}, \sigma, t) \models \bigvee_{\varphi \in \Phi} K_a E_{\mathcal{A}} \varphi$. We denote the \emph{current mutually-known primitive knowledge} of $a$ at a time $t$ as $\varphi^*_{(a,\sigma,t)} \in \Phi$ such that $\mathcal{I}, \sigma, t \models K_a E_{\mathcal{A}}\varphi^*_{(a,\sigma,t)}$ and for any $\varphi' \in \Phi$ such that $\mathcal{I}, \sigma, t \models K_a E_{\mathcal{A}} \varphi'$ it follows that $\varphi^*_{(a,\sigma,t)} \rightarrow \varphi'$.
\end{definition}

\begin{definition}[Mutually-known primitive knowledge limit] \label{def:phistarlim}
    Let $a \in \mathcal{A}$ be an agent, and $\sigma$ a run such that $\mathcal{I}, \sigma \models \Diamond( \bigvee_{\varphi \in \Phi} K_a E_{\mathcal{A}} \varphi)$. We denote by $\varphi^*_{(a,\sigma)}$ the formula $\varphi^*_{(a,\sigma)} \in \Phi$ such that $\mathcal{I}, \sigma \models \Diamond K_a E_{\mathcal{A}} \varphi^*_{(a,\sigma)}$, and for any $\varphi' \in \Phi$ such that $\mathcal{I}, \sigma \models \Diamond K_a E_{\mathcal{A}} \varphi'$, then $\varphi^*_{(a,\sigma)} \rightarrow \varphi'$
\end{definition}

We can also provide an explicit construction for $\overline{\varphi}_{(a,\sigma,t)}$ in the following way: Let $\V_{(a,\sigma,t)} = \{ init_b(v) \; \vert \; b \in \mathcal{A}, v \in \V, (\mathcal{I}, \sigma, t) \models K_a init_b(v)\}$. It is easy to verify that $\overline{\varphi}_{(a,\sigma,t)} \leftrightarrow \bigwedge_{\varphi \in \V_{(a,\sigma,t)}} \varphi$. In the same way, we define $\V_{(a,\sigma)} = \{ init_b(v) \; \vert \; b \in \mathcal{A},  v \in \V, \exists t \in \mathbb{T}; \; (\mathcal{I}, \sigma, t) \models K_a init_b(v)\}$. We can also verify that $\overline{\varphi}_{(a,\sigma)} \leftrightarrow \bigwedge_{\varphi \in \V_{(a,\sigma)}} \varphi$

For any $\sigma \in \Sigma$, and $t\in \mathbb{T}$ where $(\mathcal{I}, \sigma, t) \models \bigvee_{\varphi \in \Phi} K_a E_{\mathcal{A}} \varphi $ can also provide an explicit construction for $\varphi^*_{(a,\sigma,t)}$ in the following way:  let $\V^*_{(a,\sigma,t)} = \{ init_b (v) \; \vert \; b\in \mathcal{A}, v \in \V; (\mathcal{I}, \sigma, t) \models K_a E_{\mathcal{A}} init_b(v) \}$. We can verify that $\varphi^*_{(a,\sigma,t)} \leftrightarrow \bigwedge_{\varphi \in \V^*_{(a,\sigma,t)}} \varphi$. Again, for any $\sigma \in \Sigma$ such that $\mathcal{I}, \sigma \models \Diamond \bigvee_{\varphi \in \Phi} K_a E_{\mathcal{A}}, \varphi $ we can provide an explicit construction of $\varphi^*_{(a,\sigma)}$ in the following way: let $\V^*_{(a,\sigma)} = \{ init_b (v) \; \vert \; b\in \mathcal{A}, v \in \V; \mathcal{I}, \sigma\models \Diamond K_a E_{\mathcal{A}} init_b(v) \}$. We can verify that $\varphi^*_{(a,\sigma)} \leftrightarrow \bigwedge_{\varphi \in \V^*_{(a,\sigma)}} \varphi$.

The formulas in the following definition shape the distributed system upon which we build our epistemic solution to stabilizing tasks. For convenience, we define $\dec_a(v)$ as \emph{forever choosing} the value v, i.e.,
$\dec_a(v) := \Box \cho_a(v)$.

\begin{definition}[Stable choice system]
\label{def:stabchoice}
We say that a model is a stable choice system if the following properties hold:

\begin{itemize}
    \item \textit{Stable Choice:} Each agent eventually decides some value\footnote{The decide action refers only to $\Box \cho_a(v)$, and should not be confused with the decision action in terminating tasks.}
    \begin{equation}
        \mathcal{I} \models \bigwedge_{a \in \agents} \bigvee_{v \in \V} \Diamond \dec_a(v)
    \end{equation}
    \item \textit{Choice Determinism:} each agent is allowed to choose at most a single value at any point in the run.
    \begin{equation} \label{eq:ChoiceDeterminism}
        \mathcal{I} \models \bigwedge_{a \in \agents, v \in \V} \bigl( \cho_a(v) \rightarrow
            \bigwedge_{w \in \V \setminus \{ v \}} \neg \cho_a(w) \bigr)
    \end{equation}

\item \textit{Local state introspection:} agents know their own initial value,
    \begin{equation}
        \mathcal{I} \models \bigwedge_{a \in \agents} [init_a(v) \rightarrow
            K_a init_a(v) ]
    \end{equation}
    \textit{and agents have exactly one input value.}
    \begin{equation}
        \mathcal{I} \models \bigwedge_{a \in \agents} (\bigvee_{v\in\V} init_a(v) \bigwedge_{w \in \V\setminus\{ v \}} \neg init_a(w)).
    \end{equation}
\item \textit{Perfect input recall:} agents can only increase their knowledge about input values.
    \begin{equation}\label{eq:perfectrecall}
        \mathcal{I} \models
            \bigwedge_{a \in \agents, \varphi \in \Phi} K_{a} \varphi \rightarrow \Box K_a \varphi.
    \end{equation}
\end{itemize}

\end{definition}

We proceed to define stable choice systems that are consistent with stabilizing agreement.

\begin{definition}[Stabilizing agreement]
\label{def:Ef_ST}
    We say that a \emph{stable choice system} is consistent with stabilizing agreement if \emph{(Agreement)} and \emph{(Validity)} hold.
\begin{itemize}
    \item (Agreement) : There is a value such that every agent decides on that value
    \begin{equation}
         \mathcal{I} \models \bigvee_{v \in \mathcal{V}} \bigwedge_{a \in \agents} \Diamond \dec_a(v)
    \end{equation}
    \item (Validity) : An agent can only choose a known initial value of some agent
     \begin{equation}
          \mathcal{I} \models \bigwedge_{v \in \V} \bigl( \cho_a(v) \rightarrow   K_a \bigvee_{b \in \agents} init_b (v) \bigr)
     \end{equation}
\end{itemize}
    
\end{definition}

\section{An Epistemic Solution for Stabilizing Agreement}
\label{sec:results}
In this section we introduce additional conditions that enable a system to be consistent with stabilizing agreement. Our first condition, the \emph{Second Depth Broadcaster condition}, is a knowledge liveness condition that guarantees that at any run, some agent is able to lift its primitive knowledge to eventual second depth mutual knowledge. The second condition, the \emph{Largest Mutually-Known Choice condition} imposes a restriction on the choose action of an agent, and it implicitly implies a precondition for choosing a value. The \emph{Largest Mutually-Known Choice condition} can be stated informally as follows: as soon as an agent $a$ achieves some knowledge of mutually-known primitive knowledge, then it must choose a value according to a deterministic and a-priori commonly known rule from the pool of mutually-known values.

\begin{definition}[The Second Depth Broadcaster Condition]\label{def:secdepthbr}
    We say that a model is consistent with \emph{The Second Depth Broadcaster Condition} if there is an agent $a \in \mathcal{A}$ that is infinitely often capable of lifting its primitive input knowledge to eventual second depth mutual knowledge. More precisely:

    \[ \mathcal{I} \models \bigvee_{a \in \agents} \bigwedge_{\varphi \in \Phi}
            \bigl( K_a \varphi \rightarrow \Diamond E_{\agents} E_{\agents} \varphi \bigr) \]    

\end{definition}

In addition to this new condition, we also assume a pre-determined \emph{value selection strategy}. Intuitively, this strategy is an a-priori commonly known method for selecting consistently a value from a pool of values. A popular selection for a choice strategy is usually selecting either the minimum or the maximum value from the set. 

\begin{definition}[Value Selection Strategy] \label{def:valselstrat}
Let $f: \Phi \rightarrow \V$ be a function. We say that $f$ is a \emph{value selection strategy} iff for any $\varphi \in \Phi$:
\[\mathcal{I} \models \varphi \rightarrow (\bigvee_{a \in \mathcal{A}}init_a(f(\varphi))). \]
\end{definition}

The epistemic condition for choosing a value is captured by the following definition:
\begin{definition}[The Largest Mutually-Known Choice Condition]
\label{def:Largest_E_CC}
    We say that a system is consistent with the largest mutually known choice condition if the following holds for some value selection strategy $f$:
    \[ \mathcal{I} \models \bigwedge_{a\in\agents}(\bigvee_{\varphi \in \Phi} K_a E_{\mathcal{A}} \varphi) \rightarrow \cho_a(f(\varphi^*_{(a,\sigma,t)}))\]
\end{definition}
Intuitively, for the sake of the stabilization agreement, it is safe for an agent $a$ to choose among the largest set of input values that $a$ knows to be mutual knowledge among all agents.

The following theorem says that an interpreted system satisfying the system assumptions above solves the stabilizing agreement task:
\begin{theorem}

Let $\mathcal{I}$ be an interpreted system consistent with \emph{choice determinism}, \emph{local state introspection}, \emph{perfect input recall}, \emph{second depth broadcaster} and \emph{largest mutually-known condition}, then $\mathcal{I}$ is consistent with stabilizing agreement.
\end{theorem}

\begin{proof}
Let $\sigma \in \Sigma$, since $\mathcal{I} \models \bigvee_{a \in \agents} \bigwedge_{\varphi \in \Phi}
            \bigl( K_a \varphi \rightarrow \Diamond E_{\agents} E_{\agents} \varphi \bigr)$, then there exists an agent $a \in \mathcal{A}$ such that $\mathcal{I}, \sigma \models \bigwedge_{\varphi \in \Phi}
            \bigl( K_a \varphi \rightarrow \Diamond E_{\agents} E_{\agents} \varphi \bigr)$. In particular consider $\overline{\varphi}_{(a,\sigma)} \in \Phi$. From the definition of $\overline{\varphi}_{(a,\sigma)}$, it follows that $\mathcal{I}, \sigma \models \Diamond K_a \overline{\varphi}_{(a,\sigma)}$. 

            Therefore, there exists a time $t_1 \in \mathbb{T}$ such that $(\mathcal{I}, \sigma, t_1) \models K_a \overline{\varphi}_{(a,\sigma)}$. It follows that $\mathcal{I}, \sigma, t_1 \models \Diamond E_{\mathcal{A}} E_{\mathcal{A}} \overline{\varphi}_{(a,\sigma)}$. 

            In particular, there exists a time $t_2 \geq t_1 \in \mathbb{T}$ such that $(\mathcal{I}, \sigma, t_2) \models E_{\mathcal{A}} E_{\mathcal{A}} \overline{\varphi}_{(a,\sigma)}$.

            Now consider $t'$ any arbitrary time $t' \geq t_2$. Since we assume knowledge to be stable, $(\mathcal{I}, \sigma, t') \models E_{\mathcal{A}} E_{\mathcal{A}} \overline{\varphi}_{(a,\sigma)}$. It follows that $\varphi^*_{(b,\sigma,t')}$ is well defined (modulo permutation) for any agent $b \in \mathcal{A}$. 

            Consider an arbitrary agent $b \in \mathcal{A}$. From the definition of $\varphi^*_{(b,\sigma,t')}$, it follows that $(\mathcal{I}, \sigma, t') \models K_b E_{\mathcal{A}} \varphi^*_{(b,\sigma,t')}$. In particular $(\mathcal{I}, \sigma, t') \models K_b K_a \varphi^*_{(b,\sigma,t')}$. From knowledge factivity, $(\mathcal{I}, \sigma, t') \models K_a \varphi^*_{(b,\sigma,t')}$. From the definition of $\overline{\varphi}_{(a,\sigma)}$, it follows that $\overline{\varphi}_{(a,\sigma)} \rightarrow \varphi^*_{(b,\sigma,t')}$.

            On the other hand, since $(\mathcal{I}, \sigma, t') \models E_{\mathcal{A}} E_{\mathcal{A}} \overline{\varphi}_{(a,\sigma)}$, it follows that $(\mathcal{I}, \sigma, t') \models K_b E_{\mathcal{A}}\overline{\varphi}_{(a,\sigma)}$. It follows from the definition of $\varphi^*_{(b,\sigma,t')}$ that $\varphi^*_{(b,\sigma,t')} \rightarrow \overline{\varphi}_{(a,\sigma)}$. Therefore, for any arbitrary time $t' \geq t_2$ and any arbitrary agent $b \in \mathcal{A}$, $\overline{\varphi}_{(a,\sigma)} \leftrightarrow \varphi^*_{(b,\sigma,t')}$.

            Note that $(\mathcal{I}, \sigma, t') \models \bigvee_{\varphi \in \Phi} K_b E_{\mathcal{A}} \varphi$, for any $b\in \mathcal{A}$ since $(\mathcal{I}, \sigma, t') \models K_b E_{\mathcal{A}}\overline{\varphi}_{(a,\sigma)}$ and $\overline{\varphi}_{(a,\sigma)} \in \Phi$. Since the largest mutually-known choice condition holds by assumption, then $(\mathcal{I}, \sigma, t') \models \cho_b(f(\varphi^*_{b,\sigma,t'}))$. Therefore $(\mathcal{I}, \sigma, t') \models \cho_b(f (\overline{\varphi}_{(a,\sigma)}))$. It follows that $(\mathcal{I}, \sigma, t_2) \models \Box \cho_b(f (\overline{\varphi}_{(a,\sigma)}))$. Therefore $\mathcal{I}, \sigma \models \Diamond \Box \cho_b(f (\overline{\varphi}_{(a,\sigma)}))$ for any arbitrary $b \in \mathcal{A}$. It follows that $\mathcal{I}, \sigma \models \bigwedge_{b \in \mathcal{A}} \cho_b(f (\overline{\varphi}_{(a,\sigma)}))$.
            
            Finally, note that $\overline{\varphi}_{(a,\sigma)}$ is a fixed element of $\Phi$ per each run $\sigma \in \Sigma$, and $f$ is a value selection strategy; therefore $f(\overline{\varphi}_{(a,\sigma)}) \in \mathcal{V}$. This implies that $\mathcal{I} \models \bigvee_{v \in \mathcal{V}} \bigwedge_{a \in \agents} \Diamond \dec_a(v)$.

            Validity follows from the fact that processes only choose values from their current primitive knowledge.
\end{proof}

\section{Conclusion and Further Work}

We presented a first ever sufficient epistemic characterization for stabilizing agreement and prove it correct. In \Cref{def:secdepthbr} we presented a sufficient condition on the knowledge to solve stabilizing agreement, shedding new light on the epistemic understanding of non-terminating tasks. 
After identifying sufficient conditions, the necessary conditions for stabilizing agreement are an obvious next direction.
Further, our result that eventual second-order mutual knowledge suffices hints proximity to the notion of super-experts in gossip protocols~\cite{vdgrsuergossip}, and connections to higher-order eventual group knowledge, such as the eventual common hope for fault-tolerant Firing Rebels with relay~\cite{fire} can be explored.

\bibliographystyle{splncs04}
\bibliography{references}
\end{document}